\documentclass[onecolumn,amsmath,amssymb,aps,pra,11pt]{revtex4}

\usepackage{graphicx}
\usepackage{latexsym}
\usepackage{amsmath,amssymb}

\usepackage{epsfig}
\usepackage{graphicx}

\usepackage{graphicx}
\usepackage{amssymb,amsmath,times}

\newtheorem{theorem}{Theorem}
\newtheorem{lemma}{Lemma}
\newtheorem{defn}{Definition}

\newenvironment{proof}[1][Proof]{\noindent\textbf{#1.} }{\ \rule{0.5em}{0.5em}}

\def\re{\mathrm{Re}}
\def\im{\mathrm{Im}}
\def\erf{\mathrm{erf}}
\def\tr{\mathrm{Tr}}
\def\halfp{{\mbox{$+\frac{1}{2}$}}}
\def\halfm{{\mbox{$-\frac{1}{2}$}}}
\def\iint{\int\int}

\newcommand{\gr}[1]{\boldsymbol{#1}}

\usepackage{color}

\begin{document}

\title[Quantum tomography and nonlocality]{Quantum tomography and nonlocality}

\author{Evgeny V. Shchukin}
\email{evgeny.shchukin@gmail.com}
\address{Institute of Physics, Johannes-Gutenberg University of Mainz, Staudingerweg 7, 55128 Mainz, Germany}
\author{Stefano Mancini}
\email{stefano.mancini@unicam.it}
\address{School of Science and Technology, University of Camerino,  62032 Camerino, Italy\\
\& INFN Sezione di Perugia, I-06123 Perugia, Italy}

\begin{abstract}
We present a tomographic approach to the study of quantum nonlocality in 
multipartite systems. Bell inequalities for tomograms belonging to a generic tomographic scheme 
are derived by exploiting tools from convex geometry. 
Then, possible violations of these 
inequalities are discussed in specific tomographic realizations providing 
some explicit examples.
\end{abstract}


\maketitle

\section{Introduction}
The Bell inequalities \cite{ph-1-195} demonstrate paradigmatic difference of 
quantum and classical worlds. They were originally written for 
dichotomic (spin$-\frac{1}{2}$) variables \cite{BO}. Spin$-\frac{1}{2}$ 
operators realize the Lie algebra of the $\mathrm{SU}(2)$ group. For several 
spin particles their spin operators form Lie algebra of the tensor product of 
the Lie algebras. Due to algebraic equivalence of the operators satisfying 
commutation relations of the Lie algebra constructed from particle spin 
operators and constructed from creation and annihilation operators of a 
field, one can obtain Bell inequalities also for the case of continuous 
variables besides discrete ones \cite{CV}. Beyond the specific operators 
involved in the Bell inequalities, their possible violations obviously depend on 
the state under consideration.

For a (multipartite) classical system with fluctuations, the system state is 
described by means of a joint probability distribution function of random 
variables corresponding to the subsystems. In contrast, for a (multipartite) quantum system the 
state is described by the density matrix. In view of this difference the 
calculations of the system's statistical properties (including correlations) are 
accomplished differently in classical and quantum domains.

Recently, a probability representation of quantum mechanics has been suggested 
\cite{FOUND}. This representation, equivalent to all other well known 
formulations of quantum mechanics (see, e.g. \cite{STYER}), goes back to 
\emph{quantum tomography}, a technique used for quantum state reconstruction 
\cite{REVIEW}. The approach makes use of a set of fair probabilities, 
\emph{tomograms}, to ``replace'' the notion of quantum state.
It has also been understood \cite{JRLR} that for classical statistical mechanics 
the states with fluctuations can be described as well by tomograms related to 
standard probability distributions in classical phase-space. A comparison of 
classical and quantum tomograms can be found in Ref.\cite{JRLR, PHYSD}.

Thus, in the probability representation, tomograms turned out to be a unique 
tool to describe both classical and quantum states. As a consequence they 
represent a natural bed where to place inequalities marking the boarder-line 
between quantum and classical worlds. 
Tomograms can be either  continuous or discrete variable functions depending on 
the tomographic scheme (realization). In both cases they might be directly 
used to test nonlocality. This possibility was described for symplectic 
tomography \cite{SYMTOM} in bipartite system \cite{job-5-S333}
and spin tomography \cite{SPINTOM} still in bipartite system \cite{ninni}.

Here we shall derive Bell inequalities for \emph{multipartite} systems in terms of 
tomograms belonging to a \emph{generic} tomographic scheme. Then, we shall discuss the 
possibility to violate such inequalities depending on the tomographic 
realization.

The layout of the paper is the following. In Section \ref{qtom} we formalize 
quantum tomography in a multipartite setting. Then, in Section \ref{belllike}
we derive the Bell inequalities in terms of tomograms. In Section \ref{qviol}
we provide some evidences of violations of such inequalities for spin$-\frac{1}{2}$ systems as well
as for field modes and
finally draw the conclusions in Section \ref{conclu}.

\section{Quantum tomography}
\label{qtom}

Here we briefly review the general quantum tomography approach for a single system,
by detailing three relevant cases (optical \cite{OPTTOM}, spin \cite{SPINTOM} and photon-number 
tomography \cite{PNT}) and then extend the formalism to multipartite systems.

The basic ingredients of any tomographic scheme are a Hilbert space 
$\mathcal{H}$ associated with space of the system under consideration and a pair 
of measurable sets $(X, \Lambda)$ with measures $\mu(x)$ and $\nu(\lambda)$ 
correspondingly. More precisely, the set of system states is the set 
$\mathcal{S}(\mathcal{H})$ of Hermitian non-negative trace-class operators on 
$\mathcal{H}$ with trace $1$. Usually the set $X$ is the spectrum of an 
observable of the system and the set $\Lambda$ plays the role of 
transformations. 

We use the notation $\mathcal{P}(X)$ for the set of probability distributions on 
$X$, i.e. the set of nonnegative measurable functions
$p: X \to \mathbb{R}$
normalized to one in the following sense $\int p(x)\,d\mu(x) = 1$.

Both sets $\mathcal{S}(\mathcal{H})$ and $\mathcal{P}(X)$ are closed with respect 
to the convex combinations: if $\hat{\varrho}, \hat{\sigma} \in 
\mathcal{S}(\mathcal{H})$ (resp. $p(x), q(x) \in \mathcal{P}(X)$) and $a \in [0, 1]$ 
then
\begin{equation*}
    a \hat{\varrho} + (1-a) \hat{\sigma} \in \mathcal{S}(\mathcal{H})\quad
    ({\rm resp.}\; a p(x) + (1-a) q(x) \in \mathcal{P}(X)).
\end{equation*}

\begin{defn}\label{tomomap}
A map
$\mathcal{T}: \mathcal{S}(\mathcal{H}) \to \mathbb{R}^{X \times \Lambda}$
is called tomographic map if the following three conditions are satisfied: 
\begin{enumerate}
\item 
for any $\hat{\varrho} \in \mathcal{S}(\mathcal{H})$ the image 
$\mathcal{T}(\hat{\varrho}): X \times \Lambda \to \mathbb{R}$ restricted on the 
set $X \times \{\lambda\}$ is a probability density on $X$ 
\begin{equation*}
    \mathcal{T}_\lambda(\hat{\varrho}) \in \mathcal{P}(X) \quad \forall \lambda \in \Lambda,
    \quad{\rm where}\quad
    \mathcal{T}_\lambda(\hat{\varrho}) = \mathcal{T}(\hat{\varrho})|_{X \times \{\lambda\}}: X \to \mathbb{R}.
\end{equation*}
\item
the map $\mathcal{T}$ preserves convex combinations
\begin{equation*}
    \mathcal{T}(a \hat{\varrho} + (1-a) \hat{\sigma}) = a \mathcal{T}(\hat{\varrho}) + (1-a) \mathcal{T}(\hat{\sigma}),
    \quad
    \forall\hat{\varrho}, \hat{\sigma} \in \mathcal{S}(\mathcal{H}), a \in [0, 1].
\end{equation*}
\item
the map $\mathcal{T}$ is one-to-one
\begin{equation*}
    \mathcal{T}(\hat{\varrho}) = \mathcal{T}(\hat{\sigma}) 
    \Leftrightarrow
    \hat{\varrho} = \hat{\sigma}.
\end{equation*}
\end{enumerate}
\end{defn}
These conditions have simple meaning: (i) means that the tomogram 
$\mathcal{T}(\hat{\varrho})$ of any state $\hat{\varrho}$ is a probability 
distribution on $X$ parameterized by the points of $\Lambda$, (ii) is the 
linearity condition, and (iii) requires that the tomogram of each state be 
unique, or, in other words, that any state can be unambiguous reconstructed from 
its tomogram.

In the present work we deal with tomographic maps of the following form
\begin{equation}\label{calT}
    \mathcal{T}(\hat{\varrho})(x, \lambda) \equiv p_{\hat{\varrho}}(x, \lambda) = \tr\Bigl(\hat{\varrho}\hat{U}(x,
\lambda)\Bigr),
\end{equation}
where $\hat{U}(x, \lambda)$ is a family of operators on $\mathcal{H}$ 
parameterized by points $(x, \lambda)$ of the set $X \times \Lambda$. In the 
examples considered below the state $\hat{\varrho}$ can be reconstructed from 
its tomogram $p_{\hat{\varrho}}(x, \lambda)$ according to the formula
\begin{equation}\label{eq:D}
    \hat{\varrho} = \iint_{X \times \Lambda} p(x, \lambda)
    \hat{\cal D}(x, \lambda)\,d\mu(x)\,d\nu(\lambda),
\end{equation}
for the appropriate $(x, \lambda)$-parameterized family of operators $\hat{\cal 
D}(x, \lambda)$ on $\mathcal{H}$.

The set $X$ is the spectrum of an observable $\hat{O}$ and the set $\Lambda$ is a group 
equipped with a representation (in general projective)
$ \pi: \Lambda \to \mathcal{H}$
in $\mathcal{H}$. The operators $\hat{U}(x, \lambda)$ have the following form
\begin{equation}\label{eq:Ugen}
    \hat{U}(x, \lambda) = \pi(\lambda)|x\rangle\langle x|\pi^\dagger(\lambda),
\end{equation}
where $|x\rangle$ is an eigenstate of the observable $\hat{O}$.
For a group theoretical approach to quantum tomography see \cite{jmp-41-7940}.
See also \cite{Ibort} for a relation to grupoids.


\subsection{Spin tomography}
\label{spintom}

Let us consider a system with spin $j$. In this case we have: $\mathcal{H} = 
\mathbb{C}^{2j+1}$, $X = \{-j, -j+1, \ldots, j-1, j\}$ and $\Lambda = {\rm SO}(3, 
\mathbb{R})$. We denote the elements of the sets $X$ and $\Lambda$ as $s$ and 
$\Omega$ respectively. The measure on $X$ is equal to one on each element, so 
the corresponding integral is simply the finite sum over $2j+1$ terms. The 
measure on ${\rm SO}(3, \mathbb{R})$ is Haar's one. For the group ${\rm SO}(3, \mathbb{R})$, 
parameterized with Euler angles $\Omega \equiv (\varphi, \psi, \theta)$ 
the measure $\nu(\Omega)$ reads
$\nu(\Omega) \equiv \nu(\varphi, \psi, \theta) = 
\sin\psi\,d\varphi\,d\psi\,d\theta$ and
the operator $\hat{U}$ of (\ref{eq:Ugen}) takes the form
\begin{equation}
    \hat{U}(s, \Omega) = \hat{K}(\Omega)|j, s\rangle\langle j, 
s|\hat{K}^\dagger(\Omega).
\end{equation}
Here the vectors $|j, s\rangle$, $s = -j, -j+1, \ldots, j-1, j$ are the basis of 
the space $\mathbb{C}^{2j+1}$ (eigenvectors of the spin projection $\hat{s}_z$) 
and the operators $\hat{K}(\Omega)$ are the operators of the irreducible 
representation of ${\rm SO}(3, \mathbb{R})$ in $\mathbb{C}^{2j+1}$. Their matrix 
elements are given by 
\begin{eqnarray}
    \langle j, s|\hat{K}(\Omega)|j, s^\prime\rangle& = &
e^{i(s\theta+s^\prime\varphi)}
    \sqrt{\frac{(j+s^\prime)!(j-s^\prime)!}{(j+s)!(j-s)!}} \nonumber \\
    &\times&\cos^{s+s^\prime}(\psi/2)\sin^{s^\prime-s}(\psi/2)
    P^{(s^\prime-s, s^\prime+s)}_{j-s^\prime}(\cos\psi),
\end{eqnarray}
with $P^{(\alpha, \beta)}_n(x)$ the Jacobi polynomials.

Then the tomogram $p(s, \Omega) \equiv p(s, \varphi, \psi, 
\theta)$ of (\ref{calT}) is
\begin{equation}\label{eq:ps}
    p(s, \Omega) = \langle j, 
s|\hat{K}(\Omega)\hat{\varrho}\hat{K}^\dagger(\Omega)|j, s\rangle.
\end{equation}
Due to the property
$\langle j, s|\hat{K}(\Omega)|j, s^\prime\rangle =
    (-1)^{s^\prime-s}\langle j, -s|\hat{K}(\Omega)|j, -s^\prime\rangle$,
the tomogram does not depend on the angle $\theta$, i.e. $p(s, \varphi, \psi, 
\theta) \equiv p(s, \varphi, \psi)$.

Finally, the operator $\hat{\cal D}$ of (\ref{eq:D}) results
\begin{equation*}
    \hat{\cal D}(s, \Omega) = \sum^j_{n, m = -j}\langle j, n|
    \hat{\cal D}(s, \Omega)|j, m\rangle |j, n\rangle\langle j, m|,
\end{equation*}
where the matrix elements $\langle j, n|\hat{\cal D}(s, \Omega)|j, m\rangle$ are 
given by the following expression
\begin{eqnarray}\label{eq:Dm}
    \langle j, n|\hat{\cal D}(s, \Omega)|j, m\rangle &=&
    \frac{(-1)^{s+m}}{8\pi^2}\sum^{2j}_{j_3 = 0} (2j_3+1)^2 \nonumber\\
    &\times& \sum^{j_3}_{k = -j_3}\langle j, k|\hat{K}(\Omega)|j, 0\rangle
\left(\begin{array}{ccc}
    j & j & j_3 \\
    n & -m & k
\end{array}\right)
\left(\begin{array}{ccc}
    j & j & j_3 \\
    s & -s & k
\end{array}\right),
\end{eqnarray}
in terms of Wigner $3j$-symbols.


\subsection{Optical tomography}
\label{opttom}

Here we have : $\mathcal{H} = L_2(\mathbb{R})$, $X = \mathbb{R}$ and $\Lambda = 
\{e^{i\theta}|\theta \in [0, 2\pi]\}$. The measures on $X$ and $\Lambda$ are 
Lebegue's ones. The operator corresponding to Eq.(\ref{eq:Ugen}) reads
\begin{equation}\label{eq:Uo}
    \hat{U}(X, \theta) = \hat{R}(\theta)|X\rangle\langle 
X|\hat{R}^\dagger(\theta),
\end{equation}
where $\hat{R}(\theta)$ is the rotation operator
\begin{equation*}
    \hat{R}(\theta) = 
\exp\left(i\frac{\theta}{2}(\hat{x}^2+\hat{p}^2)\right),
\end{equation*}
acting on and the canonical position $\hat{x}$  and momentum  $\hat{p}$ operators as
\begin{equation*}
\hat{R}(\theta)
\left(\begin{array}{c}
    \hat{x} \\
    \hat{p}
\end{array}\right)
\hat{R}^\dagger(\theta) =
\left(
\begin{array}{cc}
    \cos\theta & -\sin\theta \\
    \sin\theta & \cos\theta
\end{array}\right)
\left(\begin{array}{c}
    \hat{x} \\
    \hat{p}
\end{array}\right).
\end{equation*}
In other words, $\hat{U}(X, \theta)$ of (\ref{eq:Uo}) is the projector on the 
rotated eigenvector $|X\rangle$ of the position operator $\hat{x}$. The tomogram 
$p(X, \theta)$ of (\ref{calT}) is the diagonal matrix element 
\begin{equation}
\label{eq:opttom}
    p(X, \theta) = \langle 
X|\hat{R}(\theta)\hat{\varrho}\hat{R}^\dagger(\theta)|X\rangle.
\end{equation}
Furthermore, the operator $\hat{\cal D}$ of (\ref{eq:D}) results
\begin{eqnarray*}
    \hat{\cal D}(X, \theta) 
    =\frac{1}{4\pi}\int |r| 
\exp\Bigl(-ir(X-\cos\theta\hat{x}-\sin\theta\hat{p})\Bigr)\,dr.
\end{eqnarray*}


\subsection{Photon-Number tomography}
\label{pntom}

Here we have: $\mathcal{H} = L_2(\mathbb{R})$, $X = \mathbb{Z}_+ = \{0, 1, 
\ldots\}$ and $\Lambda = \mathbb{C}$. We denote the elements of the sets $X$ and 
$\mathbb{C}$ as $n$ and $\alpha$ respectively. The measure on $X$ is equal to 
one on each element and the measure on $\mathbb{C}$ is $(1/\pi)d^2\alpha$, where 
$d^2\alpha = d\re\alpha\,d\im\alpha$ is the Lebegue's measure on the real plane. 
Here, the operator $\hat{U}$ is the projector onto the displaced Fock state
\begin{equation}
    \hat{U}(n, \alpha) = \hat{D}(\alpha)|n\rangle\langle 
n|\hat{D}^\dagger(\alpha),
\end{equation}
with
\begin{equation*}
D(\alpha)\equiv\exp\left[\frac{\alpha-\alpha^*}{\sqrt{2}}\hat{x}
-i\frac{\alpha+\alpha^*}{\sqrt{2}}\hat{p}\right].
\end{equation*}
From (\ref{calT}) the tomogram $p(n, \alpha)$ reads 
\begin{equation}
\label{eq:pntom}
    p(n, \alpha) = \langle n|\hat{D}(\alpha)\hat{\varrho}\hat{D}^\dagger(\alpha)|n\rangle.
\end{equation}
Furthermore, the operator $\hat{\cal D}$ of (\ref{eq:D}) becomes in this case
\begin{eqnarray*}
    \hat{\cal D}(n, \alpha) 
=4(-1)^n\sum^{+\infty}_{m = 0} (-1)^m\hat{D}(\alpha)|m\rangle\langle 
m|\hat{D}^\dagger(\alpha).
\end{eqnarray*}


\subsection{Tomography for multi-partite systems}
\label{tommulti}

The generalization for multi-partite systems is straightforward. 

\begin{defn}
Consider a $n$-partite system with the state space $\mathcal{H}^{\otimes n}$ 
and $n$ tomographic schemes, one for each part 
with sets $(X_k, \Lambda_k)$ and operators $\hat{U}_k(x_k, \lambda_k)$ and 
$\hat{\cal D}_k(x_k, \lambda_k)$, $k = 1, \ldots, n$. The tomographic scheme for 
the whole system is then constructed as the direct product of these schemes, by 
using
\begin{eqnarray}
&&    X \equiv \prod^n_{k=1}X_k, \quad \Lambda \equiv \prod^n_{k=1}\Lambda_k,
\nonumber\\
&&    \hat{U}({\gr{x}}, \lambda) \equiv \bigotimes^n_{k=1}\hat{U}_k(x_k, 
\lambda_k), \quad
    \hat{\cal D}({\gr{x}}, \lambda) \equiv \bigotimes^n_{k=1}\hat{D}_k(x_k, 
\lambda_k),\label{eq:U2}
\end{eqnarray}
where ${\gr{x}} \equiv (x_1, \ldots, x_n)$, ${\gr\lambda} \equiv (\lambda_1, 
\ldots, \lambda_n)$ and the measures $\mu(\gr{x})$, $\nu(\gr{\lambda})$ on $X$, 
$\Lambda$ are direct products of $\mu_1(x_1), \ldots, \mu_n(x_n)$ and 
$\nu_1(\lambda_1), \ldots, \nu_n(\lambda_n)$ respectively. The tomogram 
$p(\gr{x}, \gr{\lambda})$ of a state $\hat{\varrho}$ (generalizing 
(\ref{calT})) is
\begin{equation}
    p(\gr{x}, \gr{\lambda}) = \tr\Bigl(\hat{\varrho}
\hat{U}(\gr{x}, \gr{\lambda})\Bigr).
\end{equation}
For any $\gr{\lambda} \in \Lambda$ it is a probability distribution on $X$, thus
$    \int_X p(\gr{x}, \gr{\lambda})\, d\mu(\gr{x}) = 1$.
\end{defn}

\bigskip

\noindent
\textbf{Remark.}
From the definition (\ref{eq:U2}) of the operator $\hat{U}(\gr{x}, 
\gr{\lambda})$ it immediately follows that the tomogram $p(\gr{x}, 
\gr{\lambda})$ of a factorized state
\begin{equation}\label{eq:rhof}
    \hat{\varrho} = \hat{\varrho}_1 \otimes \ldots \otimes 
\hat{\varrho}_n
\end{equation}
is also factorized, i.e.
\begin{equation}\label{eq:pf}
    p(\gr{x}, \gr{\lambda}) = p_1(x_1, \lambda_1) \ldots p_n(x_n, 
\lambda_n),
\end{equation}
where $p_k(x_k, \lambda_k)$ is the tomogram of the state $\hat{\varrho}_k$. More 
generally, the tomogram of a separable state
\begin{equation}\label{eq:rhos}
    \hat{\varrho} = \sum^{+\infty}_{i=0} a_i \hat{\varrho}^{(i)}_1 
\otimes \ldots
    \hat{\varrho}^{(i)}_n, \quad\quad a_i \geqslant 0, \quad 
\sum^{+\infty}_{i=0} a_i = 1
\end{equation}
is also separable in the following sense
\begin{equation}
    p(\gr{x}, \gr{\lambda}) = \sum^{+\infty}_{i=0} a_i 
p^{(i)}_1(x_1, \lambda_1) \ldots
    p^{(i)}_n(x_n, \lambda_n),
\end{equation}
where $p^{(i)}_k(x_k, \lambda_k)$ is the tomogram of the state 
$\hat{\varrho}^{(i)}_k$.


\section{Bell inequalities for tomograms}
\label{belllike}

Let us consider a $n$-partite system in tomographic representation, whith each 
subsystem supplied by a tomographic map ${\cal T}_k$, $k=1,\ldots,n$. 
The tomogram $p(\gr{x}, \gr{\lambda})$ of a state $\hat{\varrho}$ is a function 
of $2n$ arguments and with respect to one half of them it is a probability 
distribution. We will show that in general it cannot be considered as a 
classical joint probability.

\begin{defn}\label{def:setsYZ}
For any $k = 1, \ldots, n$ let $Y_k$ and $Z_k$ be two measurable sets such that
\begin{equation*}
    X_k = Y_k \bigcup Z_k, \quad Y_k \bigcap Z_k = \emptyset,
\end{equation*}
and for any $\lambda_k \in \Lambda_k$ let $A_k(\lambda_k)$ be a dichotomic 
random variable on $X = \prod_k X_k$ such that
\begin{eqnarray}\label{eq:rv}
    \mathbf{P}(A_k(\lambda_k) = 1) &= \int_{Y_k}
    \tr \Bigl(\hat{\varrho}\hat{U}_k(x_k, \lambda_k)\Bigr)\,d\mu_k(x_k), \nonumber\\
    \mathbf{P}(A_k(\lambda_k) = -1) &= \int_{Z_k} 
    \tr \Bigl(\hat{\varrho}\hat{U}_k(x_k, \lambda_k)\Bigr)\,d\mu_k(x_k).
\end{eqnarray}
\end{defn}
Symbolically the variables  $A_k(\lambda_k)$  can be written as 
\begin{equation}\label{eq:rv2}
    A_k(\lambda_k) =
    \left\{\begin{array}{cc}
        1 & {\rm if}\ x_k \in Y_k, \\
        -1 & {\rm if}\ x_k \in Z_k
    \end{array}\right.
\end{equation}
in the coordinate system deformed by the operator $\hat{U}_k(x_k, \lambda_k)$. 
The joint probability distribution of the random variables $A_1(\lambda_1), 
\ldots, A_n(\lambda_n)$, namely
\begin{equation*}
    p_{\varepsilon_1, \ldots, \varepsilon_n}(\lambda_1, \ldots, 
\lambda_n) =
    \mathbf{P}(A_1(\lambda_1) = \varepsilon_1, \ldots, A_n(\lambda_n) 
= \varepsilon_n),
\end{equation*}
where $\varepsilon_k = \pm 1$, is given by 
\begin{equation}\label{eq:jpd}
    p_{\varepsilon_1, \ldots, \varepsilon_n}(\lambda_1, \ldots, 
\lambda_n) =
    \int_{W_1} \ldots \int_{W_n} p(\gr{x}, \gr{\lambda})\,d\gr{x},
\end{equation}
with
\begin{equation*}
    W_k =
    \left\{\begin{array}{cc}
        Y_k & {\rm if}\ \varepsilon_k = 1, \\
        Z_k & {\rm if}\ \varepsilon_k = -1.
    \end{array}\right.
\end{equation*}
The correlation function of $A_1(\lambda_1), \ldots, A_n(\lambda_n)$ results
\begin{eqnarray}\label{eq:Edef}
    E(\lambda_1, \ldots, \lambda_n) 
    &\equiv& \bigl\langle A_1(\lambda_1) \ldots A_n(\lambda_n) \bigr\rangle\nonumber\\
    &=& \sum_{\varepsilon_1, \ldots, \varepsilon_n = \pm 1}
    p_{\varepsilon_1, \ldots, \varepsilon_n}(\lambda_1, \ldots, \lambda_n)
    \varepsilon_1 \ldots \varepsilon_n.
\end{eqnarray}
\begin{defn}
Let us fix two parameters $\lambda^{(1)}_k$ and $\lambda^{(2)}_k$ for $k = 1, 
\ldots, n$ and denote
\begin{equation}\label{eq:tomcor}
    E(j_1, \ldots, j_n) \equiv E(\lambda^{(j_1)}_1, \ldots, 
\lambda^{(j_n)}_n), \quad j_k = 1, 2.
\end{equation}
Since each index $j_k$ can take $2$ values independently on all the other indices, 
there are $2^n$ correlation functions (\ref{eq:tomcor}). Then we define by 
\begin{equation}\label{eq:e}
    \gr{e} \equiv \Bigl(E(j_1, \ldots, j_n)\Bigr) \in \mathbb{R}^{2^n}.
\end{equation}
the vector of these correlation functions with some order of multi-indices 
$(j_1, \ldots, j_n)$. 
\end{defn}
It is convenient to enumerate the functions $E(j_1, \ldots, j_n)$. For this 
purpose we use the binary base with ``digits" $1$ and $2$ instead of $0$ and 
$1$. This means that we use the following one-to-one correspondence 
\begin{equation*}
    \{1, \ldots, 2^n\} \ni j \leftrightarrow (j_1, \ldots, j_n),\quad j_k = 1, 2,
\end{equation*}
where $j$ and $(j_1, \ldots, j_n)$ are related to each other according to 
\begin{equation}\label{eq:j}
    j = (j_1-1)2^{n-1} + \ldots + (j_n-1)2 + j_n.
\end{equation}
By virtue of such an ordering, the vector $\gr{e}$ (\ref{eq:e}) can be written as
\begin{eqnarray}\label{eq:eo}
    \gr{e} = \Bigl(E(1), \ldots, E(2^n)\Bigr)
    =\Bigl(E(1, \ldots, 1), \ldots, E(2, \ldots, 2)\Bigr) \in 
\mathbb{R}^{2^n}.
\end{eqnarray}

What region $\Omega_n \subset \mathbb{R}^{2^n}$ fills the vector $\gr{e}$ of 
(\ref{eq:eo})? Due to the fact that each observable has only two outcomes $\pm 
1$ it follows that each correlation function (\ref{eq:tomcor}) is bounded by one 
by absolute value and, so, the set $\Omega_n$ is a subset of $2^n$-dimensional 
cube ${[-1, 1]}^{2^n}$.

Suppose that it is possible to model the result of the measurement by a random 
variable, $A_k(j_k)$, which can take two values $\pm 1$. We assume that these 
random variables can be arbitrary correlated. 
\begin{defn}\label{defn:joint}
Let us define by
\begin{eqnarray}\label{eq:p}
    p(i_1(1), &\ldots, i_n(2)) \equiv 
    \mathbf{P}\Bigl(A_1(1) = i_1(1), \ldots, A_n(2) =
    i_n(2)\Bigr),
\end{eqnarray}
the joint probability distribution for random variables $A_k(j_k)$, with 
$i_k(j_k) = \pm 1$. Since each index $i_k(j_k)$ can take independently $2$ values 
we have $2^{2n}$ numbers (\ref{eq:p}) which completely describe statistical 
characteristics of the random variables under consideration. We enumerate them 
with a single number $i = 1, \ldots, 2^{2n}$ using the same rule as for the 
correlation functions $E(j)$, namely
\begin{equation*}
    \{1, \ldots, 2^{2n}\} \ni i \leftrightarrow
    (i_1(1), \ldots, i_n(2)),
\end{equation*}
where $i$ and $(i_1(1), \ldots, i_n(2))$ are related to each other according to 
\begin{equation}\label{eq:i}
    i = (i_1(1)-1)2^{2n-1} + (i_n(1)-1)2 + i_n(2).
\end{equation}
Enumerated in such a way the probabilities (\ref{eq:p}) form a 
$2^{2n}$-dimensional vector
\begin{equation}\label{eq:po}
    \gr{p} \equiv (p_1, \ldots, p_{2^{2n}}) \in \mathbb{R}^{2^{2n}}.
\end{equation}
\end{defn}
The point (\ref{eq:po}) lies in the standard simplex
\begin{equation}\label{eq:S}
    S_{2^{2n}-1} = \left\{(p_1, \ldots, 
p_{2^{2n}})\biggm|\sum^{2^{2n}}_{i=1}p_i = 1,
    p_i \geqslant 0 \right\} \subset \mathbb{R}^{2^{2n}}.
\end{equation}

What region $\Omega_n \subset \mathbb{R}^{2^n}$ fills the vector $\gr{e}$ 
(\ref{eq:eo}) when the point $\gr{p}$ (\ref{eq:po}) runs over the simplex 
$S_{2^{2n}-1}$ (\ref{eq:S})? To answer this question we are going to explicitly 
relate $\gr{e}$ and $\gr{p}$ assuming the former expressed like classical joint 
probabilities. Then, the correlation function $E(j)$ is intended as a simple 
linear combination of $p_i$ with proper coefficients. Looking at (\ref{eq:Edef}) 
we consider such coefficients, $\mathcal{E}(j, i)$, given by the product
\begin{equation}\label{eq:Eji}
    \mathcal{E}(j, i) = i_1(j_1) \ldots i_n(j_n),
\end{equation}
where $j_k$ and $i_k(j_k)$ are ``digits" of the numbers $j$ and $i$ in the 
binary representations (\ref{eq:j}) and (\ref{eq:i}). The numbers 
$\mathcal{E}(j, i)$, $j = 1, \ldots, 2^n$, $i = 1, \ldots, 2^{2n}$ form a $2^n 
\times 2^{2n}$ matrix $\mathcal{E}_n$ and the relation between $\gr{e}$ and 
$\gr{p}$ can then be written as
\begin{equation}\label{eq:ep}
    \gr{e} = \mathcal{E}_n\gr{p}.
\end{equation}
We see that the region $\Omega_n$ is the image of the standard simplex 
$S_{2^{2n}-1}$
\begin{equation}\label{eq:Omega}
    \Omega_n = \mathcal{E}_n(S_{2^{2n}-1}),
\end{equation}
where we do not distinguish the linear map  
$\mathcal{E}_n: \mathbb{R}^{2^{2n}} \to \mathbb{R}^{2^n}$
and its matrix $\mathcal{E}_n$ in the standard bases of $\mathbb{R}^{2^{2n}}$ 
and $\mathbb{R}^{2^n}$.

Thus, we have reduced the problem of finding Bell inequalities to find the set 
$\Omega_n$. It means that  the problem of finding Bell inequalities boils down 
to a standard problem of convex geometry, referred to as convex hull problem: 
given points $\gr{c}_i$ find their convex hull, or facets of maximal dimension 
of the corresponding polytope (for notions of convex geometry see, e.g. \cite{Gruber07}).

Now we will get the Bell inequalities explicitly. Note that permutations of the 
columns of the matrix $\mathcal{E}_n$ do not change their convex hull and that 
they correspond to permutations of the components of the vector $\gr{p}$ or 
different orderings of the probabilities (\ref{eq:po}), so one can safely permute  
columns of $\mathcal{E}_n$ without altering (\ref{eq:ep}).


\begin{theorem}\label{theoBell}
The set $\Omega_n$ is specified by the (Bell) inequalities for the vector of the 
correlation functions 
\begin{equation}\label{eq:Bell}
    (\gr{e}, H_{2^n}\gr{c}) \leqslant 2^n, \quad
    \forall \gr{c} = (\pm 1, \ldots, \pm 1).
\end{equation}
The matrix $H_{2^n}$ is the Hadamard matrix recurrently defined as
\begin{eqnarray*}
H_{2^n} = \underbrace{H_2 \otimes \ldots \otimes H_2}_{n}\,,\quad
    H_2 =
    \left(\begin{array}{cc}
        1 & 1 \\
        1 & -1
    \end{array}\right).  
\end{eqnarray*}
\end{theorem}


\begin{proof}
The key fact in deriving the Bell inequality (\ref{eq:Bell}) is that the matrix 
$\mathcal{E}_n$ can be written in the following block form
\begin{equation}\label{eq:EH}
    \mathcal{E}_n =\Big(
    \underbrace{
    \begin{array}{ccccc}
        H_{2^n} & -H_{2^n} & \ldots & H_{2^n} & -H_{2^n}
    \end{array}
    }_{2^n}\Big)
\end{equation}
after appropriate arrangement of its columns. One can rewrite the r.h.s. 
of (\ref{eq:EH}) as the product of two matrices
\begin{equation*}
    \mathcal{E}_n = H_{2^n}
    \left(\begin{array}{ccccc}
        E_{2^n} & -E_{2^n} & \ldots & E_{2^n} & -E_{2^n}
    \end{array}\right) =
    H_{2^n}A_n,
\end{equation*}
which means that the linear map
$\mathcal{E}_n: \mathbb{R}^{2^{2n}} \to \mathbb{R}^{2^n}$
can be decomposed into two maps
\begin{equation*}
    \mathcal{E}_n = H_{2^n} \circ A_n, \quad
    A_n: \mathbb{R}^{2^{2n}} \to \mathbb{R}^{2^n}, \quad
    H_{2^n}: \mathbb{R}^{2^n} \to \mathbb{R}^{2^n}.
\end{equation*}
According to this decomposition (\ref{eq:ep}) reads
\begin{equation}\label{eq:eq}
    \gr{e} = H_{2^n}\gr{q},
\end{equation}
where the vector $\gr{q} = A_n \gr{p} \in \mathbb{R}^{2^n}$ is explicitly given 
by the following expression
\begin{equation}\label{eq:q}
    \gr{q} =
    \left(\begin{array}{c}
        p_1-p_{2^n+1}+\ldots-p_{(2^n-1)2^n+1} \\
        \vdots \\
        p_{2^n}-p_{2\cdot2^n}+\ldots-p_{2^{2n}}
    \end{array}\right).
\end{equation}

Define the following convex polytope 
$\mathcal{O}_N \subset \mathbb{R}^N$
\begin{equation}\label{eq:O}
    \mathcal{O}_N = \{\gr{x} \in \mathbb{R}^N | (\gr{x}, \gr{c}) 
\leqslant 1, \
    \forall \gr{c} = (\pm 1, \ldots, \pm 1) \}.
\end{equation}
As one can easily see the image of the standard simplex $S_{2^{2n}-1}$ is 
exactly the polytope $\mathcal{O}_{2^n}$, that is
$A_n(S_{2^{2n}-1}) = \mathcal{O}_{2^n}$.
From this fact we have
\begin{equation}\label{eq:OH}
    \Omega_n = H_{2^n}(\mathcal{O}_{2^n}).
\end{equation}

Now the Bell inequalities can be straightforwardly obtained from this relation. 
Just notice that a non-degenerate linear map $f: \mathbb{R}^N \to \mathbb{R}^N$ 
with the matrix $F$ maps a half-space
$\mathfrak{h} = \{\gr{x} \in \mathbb{R}^N | (\gr{x}, \gr{a}) \leqslant b \}$
to the half-space
$f(\mathfrak{h}) = \{\gr{y} \in \mathbb{R}^N | (\gr{y}, (F^T)^{-1}\gr{a}) 
\leqslant b \}$.
Taking into account the following representation of the polytope (\ref{eq:OH})
\begin{equation}\label{eq:Oc}
    \mathcal{O}_{2^n} = \bigcap_{\gr{c} = (\pm 1, \ldots, \pm 1)}
    \{\gr{q}|(\gr{q}, \gr{c}) \leqslant 1\},
\end{equation}
the symmetry of the Hadamard matrix $H_{2^n}$ and the formula for its inverse
$H^{-1}_{2^n} = \frac{1}{2^n}H_{2^n}$,
we get the explicit form of the set $\Omega_n$, i.e.
\begin{equation}
    \Omega_n = \bigcap_{\gr{c} = (\pm 1, \ldots, \pm 1)}
    \{\gr{e}|(\gr{e}, H_{2^n}\gr{c}) \leqslant 2^n\}.
\end{equation}
Hence, the Bell inequalities (\ref{eq:Bell}).
\end{proof}


\bigskip

\noindent
\textbf{Remark.}
Explicitly (\ref{eq:Bell}) can be written as 
\begin{equation}\label{eq:B}
    \left|\sum^2_{j_1, \ldots, j_n = 1} a_{j_1, \ldots, j_n} E(j_1, 
\ldots, j_n)
    \right| \leqslant 2^n,
\end{equation}
where the coefficients $a_{j_1, \ldots, j_n}$ are connected with the vector 
$\gr{c}$ by the following relation
\begin{equation}\label{eq:a}
    a_{j_1, \ldots, j_n} = \sum_{\varepsilon_1, \ldots, \varepsilon_n 
= \pm 1}
    c(\varepsilon_1, \ldots, \varepsilon_n) \varepsilon^{j_1-1}_1 
\ldots \varepsilon^{j_n-1}_n.
\end{equation}
The number $c(\varepsilon_1, \ldots, \varepsilon_n)$ here is the $i$-th 
component of the vector $\gr{c}$, where the binary representation of $i$ is $i = 
(\varepsilon_1 \ldots \varepsilon_n)_2$ with digits $+1$ and $-1$ instead of $0$ 
and $1$.

One can easily see that there are $2^{n+1}$ inequalities of the form
$\pm E(j_1, \ldots, j_n) \leqslant 1$.
They correspond to the functions $c(\varepsilon_1, \ldots, \varepsilon_n)$ that 
are columns of either $H_{2^n}$ or $-H_{2^n}$ and they are referred to as 
trivial inequalities.

Finally, notice that the well known CHSH inequality
\cite{CHSH} is a particular instance of (\ref{eq:B}) corresponding to $n=2$ and
\begin{equation*}
c(-1,-1)=-1,\quad c(-1,+1)=c(+1,-1)=c(+1,+1)=+1.
\end{equation*}


\begin{theorem}
Any separable state satisfies (\ref{eq:Bell}) with correlation 
functions (\ref{eq:tomcor}).
\end{theorem}


\begin{proof}
Let us start with a factorized state (\ref{eq:rhof}) whose 
tomogram (\ref{eq:pf}) is also factorized. Due to 
this the random variables $A_1(\lambda_1), \ldots, A_n(\lambda_n)$ are 
independent and the correlation function $E(\lambda_1, \ldots, \lambda_n)$ reads
\begin{equation}\label{eq:Eq}
    E(\lambda_1, \ldots, \lambda_n) = q_1(\lambda_1) \ldots 
q_n(\lambda_n)
\end{equation}
with
\begin{equation*}
    q_k(\lambda_k) = p^{(k)}_1(\lambda_k) - p^{(k)}_{-1}(\lambda_k),
\end{equation*}
where
\begin{equation*}
    p^{(k)}_{\varepsilon_k}(\lambda_k) = \mathbf{P}(A_k(\lambda_k) = 
\varepsilon_k), \quad
    \varepsilon_k = \pm 1.
\end{equation*}
Due to the fact that
\begin{equation*}
    p^{(k)}_1(\lambda_k) + p^{(k)}_{-1}(\lambda_k) = 1, \quad \forall 
k = 1, \ldots, n \quad
    \forall \lambda_k \in \Lambda_k,
\end{equation*}
it is clear that
$-1 \leqslant q_k(\lambda_k) \leqslant 1$.
The left hand side of the inequality (\ref{eq:Bell}) is a linear function of any 
$q_k(\lambda^{(j_k)}_k)$ where all the $q_1(\lambda^{(j_1)}_1), \ldots, 
q_n(\lambda^{(j_n)}_n)$, $j_k = 1, 2$ are considered as independent variables. A 
linear function defined on the convex set $[-1, 1]$ takes its maximum on a 
boundary point, $\pm 1$ in this case, and so, the left hand side of 
(\ref{eq:Bell}) is maximal if $q_k(\lambda^{(j_1)}_k) = \pm 1$, $j_k = 1, 2$, $k 
= 1, \ldots, n$. In such a case the vector $\gr{e}$ of correlation functions is 
a column of either $H_{2^n}$ or $-H_{2^n}$. Just note that due to (\ref{eq:Eq}) 
the vector $\gr{e}$ reads
\begin{equation*}
    \gr{e} =
    \left(\begin{array}{c}
        q_1(1) \\
        q_1(2)
    \end{array}\right)
    \otimes \ldots \otimes
    \left(\begin{array}{c}
        q_n(1) \\
        q_n(2)
    \end{array}\right),
\end{equation*}
where $q_k(j_k) = q_k(\lambda^{(j_k)}_k)$. That is to say, $\gr{e} = \gr{c}_i$ 
is the $i$-th column of $H_{2^n}$ or $-H_{2^n}$, then 
\begin{eqnarray}\label{eHc}
    (\gr{e}, H_{2^n}\gr{c}) &= \pm (\gr{c}_i, H_{2^n}\gr{c}) = 
\pm (H_{2^n}\gr{c}_i, \gr{c})=
    \pm (2^n \gr{e}_i, \gr{c}) = \pm 2^n \leqslant 2^n.
\end{eqnarray}
Here we used the orthogonality of the columns of $H_{2^n}$: $H_{2^n}\gr{c}_i = 
2^n \gr{e}_i$, where all the coordinates of $\gr{e}_i$ are zero except the 
$i$-th which is one. Hence, we have proved that all factorized states satisfy (\ref{eq:Bell}).

Let us now
consider a general separable state (\ref{eq:rhos}). Since the correlation function 
$E(\lambda_1, \ldots, \lambda_n)$ is a linear function of the state, the vector 
$\gr{e}$ is a linear combination of the vectors $\gr{e}^{(i)}$ corresponding to 
the states $\hat{\varrho}^{(i)} = \hat{\varrho}^{(i)}_1 \otimes \ldots \otimes 
\hat{\varrho}^{(i)}_n$, i.e.
\begin{equation*}
    \gr{e} = \sum^{+\infty}_{n=0} a_i \gr{e}^{(i)}.
\end{equation*}
As we have already shown each vector $\gr{e}^{(i)}$ satisfies all the Bell 
inequalities (\ref{eq:Bell}) or lies in the convex set $\Omega_n$. Once all the 
vectors $\gr{e}^{(i)}$ are in $\Omega_n$ so is their convex combination 
$\gr{e}$. This means that any separable state satisfies all the inequalities 
(\ref{eq:Bell}).
\end{proof}


\section{Quantum violations}
\label{qviol}

The Bell inequalities are of interest not because they are 
always valid but because they can be violated.
One can ask if there was a mistake in the proof of theorem \ref{theoBell}. The 
problem relies on the underlying hypothesis of locality when relating $\gr{e}$ 
with $\gr{p}$ in (\ref{eq:ep}). In doing so we have implicitly assumed (\ref{eq:Edef}) as a 
classical joint probability, which is not generally true at quantum level.

We follow Mermin \cite{prl-65-1838} to derive the only Bell 
inequality whose maximal quantum violation is the largest among all others. 

For an odd number $n$ of systems let us consider the following random variable
\begin{equation}\label{eq:Mn}
    M_n = \im \left[ \prod^n_{k=1}(A_k(1) + i A_k(2)) \right].
\end{equation}
Since each $A_k$ can take only values $\pm 1$, each term in this 
product is equal to $\sqrt{2}$ by absolute value. Furthermore, since $n$ is odd 
the whole product has the phase that is an integer multiplier of $\pi/4$. 
As a consequence we have
\begin{equation}\label{eq:Mer}
    |\langle M_n \rangle| \leqslant 2^{(n-1)/2}.
\end{equation}
Explicitly this inequality reads
\begin{equation}\label{eq:Mo}
    \left| \sum_{(j_1, \ldots, j_n) \in J} (-1)^{\delta(j_1, \ldots, 
j_n)}
    E(j_1, \ldots, j_n) \right| \leqslant 2^{(n-1)/2},
\end{equation}
where the sum here runs over the set of multi-indices $(j_1, \ldots, j_n)$ which 
contain an odd number of $2$
\begin{equation*}
    J = \Bigl \{(j_1, \ldots, j_n) \Bigm| |\{k|j_k=2\}| = 2l+1 \Bigr\}
\end{equation*}
and
\begin{equation*}
    \delta(j_1, \ldots, j_n) = l, \quad |\{k|j_k=2\}| = 2l+1.
\end{equation*}
Multiplied by $2^{(n+1)/2}$ the inequality (\ref{eq:Mo}) takes the form 
(\ref{eq:B}) and it is easy to show that it is a Bell inequality, i.e. there is 
a vector $\gr{c}$ that gives the coefficients of (\ref{eq:Mo}) (multiplied by 
$2^{(n+1)/2}$) according to (\ref{eq:a}).

We now consider an even number $n$. Let us denote the expression (\ref{eq:Mn}) 
as $M_n(1, 2)$ and the similar expression with the variables $A_k(1)$ and 
$A_k(2)$ swapped as $M_n(2, 1)$. Consider the following combination
\begin{eqnarray}\label{eq:Meo}
    \widetilde{M}_n = M_{n-1}(1, 2) (A_n(1)+A_n(2)) 
    + M_{n-1}(2, 1) (A_n(1)-A_n(2)).
\end{eqnarray}
Since $M_{n-1}(1, 2)$ is equal to $\pm 2^{n/2-1}$ and $A_n(j) = \pm 1$, 
we will have
\begin{equation}\label{eq:M2}
    |\langle \widetilde{M}_n \rangle| \leqslant 2^{n/2}.
\end{equation}
Using the explicit form (\ref{eq:Mo}) for the odd number $n-1$ one can write 
(\ref{eq:M2}) as
\begin{equation}\label{eq:Me}
    \left|\sum^2_{j_1, \ldots, j_n = 1} (-1)^{\tilde{\delta}(j_1, 
\ldots, j_n)}
    E(j_1, \ldots, j_n) \right| \leqslant 2^{n/2},
\end{equation}
where
\begin{eqnarray*}
    \tilde{\delta}(j_1, \ldots, j_n) =
    \left\{\begin{array}{cc}
    1 & {\rm if}\; j_n = 2 \; {\rm and} \; |\{k|j_k = 2\}|\ {\rm is \; nonzero \; and \; even}\\
    0 & {\rm otherwise } 
    \end{array}\right. .\nonumber\\
\end{eqnarray*}
One can see that it is a Bell inequality and multiplied by $2^{n/2}$ it takes 
the form (\ref{eq:B}). Furthermore, for $n=2$ Eq.(\ref{eq:Me}) exactly reduces to the 
CHSH inequality \cite{CHSH}.

Let us now see how the inequalities (\ref{eq:Mo}) or  
(\ref{eq:Me}) can be violated in different tomographic realizations starting from 
the following entangled state 
\begin{eqnarray}\label{eq:psi1}
    |\Psi\rangle &= \frac{1}{\sqrt{2}}
    \Bigl( |{\bf 0}\rangle + |{\bf 1}\rangle \Bigr),
\end{eqnarray}
where  ${\bf 0} = (0, \ldots, 0)$ and  ${\bf 1} = (1, \ldots, 1)$. 
Below, for the sake of simplicity, the focus will mainly be to $n=2,3$.


\subsection{Spin tomography}

Using the notation $|0\rangle \equiv |\halfm\rangle$, $|1\rangle \equiv 
|\halfp\rangle$ for the spin projection along $z$, the state (\ref{eq:psi1}) becomes
\begin{equation*}
    |\Psi\rangle = \frac{1}{\sqrt{2}}\left(|\halfm, \ldots, \halfm\rangle + 
    |\halfp, \ldots, \halfp\rangle\right),
\end{equation*}
whose tomogram, referring to Eq.(\ref{eq:Dm}), reads as
\begin{eqnarray}\label{eq:spintomex}
    p(s_1, \ldots, s_n, \Omega_1, \ldots, \Omega_n) = \frac{1}{2} \left|
    \prod^n_{j=1} \langle s_j|\hat{K}(\Omega_j)|\halfm\rangle + 
    \prod^n_{j=1} \langle s_j|\hat{K}(\Omega_j)|\halfp\rangle
    \right|^2.\nonumber\\
\end{eqnarray}

For $n=2$ we immediately get
\begin{eqnarray*}
    p(\halfp, \halfp, \Omega_1, \Omega_2) &=& p(\halfm, \halfm, \Omega_1, \Omega_2) \nonumber \\
    &=& \frac{1}{4} 
    (1 + \cos\psi_1 \cos\psi_2 + \sin\psi_1 \sin\psi_2 \cos(\varphi_1+\varphi_2)), 
\nonumber \\
    p(\halfp, \halfm, \Omega_1, \Omega_2) &=& p(\halfm, \halfp, \Omega_1, \Omega_2) \nonumber \\
    &=& \frac{1}{4}
    (1 - \cos\psi_1 \cos\psi_2 - \sin\psi_1 \sin\psi_2 \cos(\varphi_1+\varphi_2)),
\end{eqnarray*}
and the correlation function (\ref{eq:Edef}) becomes
\begin{eqnarray}\label{eq:E2spin}
    E(\Omega_1, \Omega_2) = \cos\psi_1 \cos\psi_2 + \sin\psi_1 \sin\psi_2    
\cos(\varphi_1+\varphi_2).
\end{eqnarray}
The Bell inequality (\ref{eq:Me}) reads in this case 
\begin{eqnarray}\label{eq:Bellspin}
    \Bigl|E(\Omega^{(1)}_1, \Omega^{(1)}_2) &+ E(\Omega^{(1)}_1, 
\Omega^{(2)}_2) + 
    E(\Omega^{(2)}_1, \Omega^{(1)}_2) - E(\Omega^{(2)}_1, 
\Omega^{(2)}_2)\Bigr| \leqslant 2,
\end{eqnarray}
for all $\Omega^{(j)}_k = (\varphi^{(j)}_k, \psi^{(j)}_k, \theta^{(j)}_k)$, $j,k = 
1, 2$. 
The maximum of the l.h.s. of (\ref{eq:Bellspin}) with (\ref{eq:E2spin}) is $2\sqrt{2}$ and is attained by taking e.g.
(the angles $\theta$ do not matter here)
\begin{eqnarray*}
    \Omega^{(1)}_1 &= (\varphi_1, -\pi/8, 0), & \quad 
    \Omega^{(1)}_2 = (-\varphi_1, \pi/8, 0), \\
    \Omega^{(2)}_1 &= (\varphi_1, 3\pi/8, 0), & \quad 
    \Omega^{(2)}_2 = (-\varphi_1, -3\pi/8, 0).
\end{eqnarray*}

In the case of $n=3$, from (\ref{eq:spintomex}), we have (not to overload the notation we omit the $\Omega$'s)
\begin{eqnarray*}
    p(\halfp, \halfp, \halfp) &= \frac{1}{8}[1 + \cos\psi_1 \cos\psi_2 + \cos\psi_1 \cos\psi_3 + 
\cos\psi_2 \cos\psi_3 \nonumber \\
    &- \sin\psi_1 \sin\psi_2 \sin\psi_3 \cos(\varphi_1 + \varphi_2 + \varphi_3)], \nonumber \\
    p(\halfp, \halfp, \halfm) &= \frac{1}{8}[1 + \cos\psi_1 \cos\psi_2 - \cos\psi_1 \cos\psi_3 - 
\cos\psi_2 \cos\psi_3 \nonumber \\
    &+ \sin\psi_1 \sin\psi_2 \sin\psi_3 \cos(\varphi_1 + \varphi_2 + \varphi_3)], \nonumber \\
    p(\halfp, \halfm, \halfp) &= \frac{1}{8}[1 - \cos\psi_1 \cos\psi_2 + \cos\psi_1 \cos\psi_3 - 
\cos\psi_2 \cos\psi_3 \nonumber \\
    &+ \sin\psi_1 \sin\psi_2 \sin\psi_3 \cos(\varphi_1 + \varphi_2 + \varphi_3)], \nonumber \\
    p(\halfp, \halfm, \halfm) &= \frac{1}{8}[1 - \cos\psi_1 \cos\psi_2 - \cos\psi_1 \cos\psi_3 + 
\cos\psi_2 \cos\psi_3 \nonumber \\
    &- \sin\psi_1 \sin\psi_2 \sin\psi_3 \cos(\varphi_1 + \varphi_2 + \varphi_3)], \nonumber \\
    p(\halfm, \halfp, \halfp) &= \frac{1}{8}[1 - \cos\psi_1 \cos\psi_2 - \cos\psi_1 \cos\psi_3 + 
\cos\psi_2 \cos\psi_3 \nonumber \\
    &+ \sin\psi_1 \sin\psi_2 \sin\psi_3 \cos(\varphi_1 + \varphi_2 + \varphi_3)], \nonumber \\
    p(\halfm, \halfp, \halfm) &= \frac{1}{8}[1 - \cos\psi_1 \cos\psi_2 + \cos\psi_1 \cos\psi_3 - 
\cos\psi_2 \cos\psi_3 \nonumber \\
    &- \sin\psi_1 \sin\psi_2 \sin\psi_3 \cos(\varphi_1 + \varphi_2 + \varphi_3)], \nonumber \\
    p(\halfm, \halfm, \halfp) &= \frac{1}{8}[1 + \cos\psi_1 \cos\psi_2 - \cos\psi_1 \cos\psi_3 - 
\cos\psi_2 \cos\psi_3 \nonumber \\
    &- \sin\psi_1 \sin\psi_2 \sin\psi_3 \cos(\varphi_1 + \varphi_2 + \varphi_3)], \nonumber \\
    p(\halfm, \halfm, \halfm) &= \frac{1}{8}[1 + \cos\psi_1 \cos\psi_2 + \cos\psi_1 \cos\psi_3 + 
\cos\psi_2 \cos\psi_3 \nonumber \\
    &+ \sin\psi_1 \sin\psi_2 \sin\psi_3 \cos(\varphi_1 + \varphi_2 + \varphi_3)].
\end{eqnarray*}
Thanks to these tomograms the correlation function (\ref{eq:Edef}) results
\begin{equation}\label{eq:E3spin}
E(\Omega_1,\Omega_2,\Omega_3)=-\sin\psi_1\sin\psi_2\sin\psi_3\cos(\varphi_1 + \varphi_2 + \varphi_3).
\end{equation}
Finally, the Bell inequality (\ref{eq:Mo}) in this case reads
\begin{eqnarray}
\left| E(\Omega^{(2)}_1,\Omega^{(1)}_2,\Omega^{(1)}_3)+ E(\Omega^{(1)}_1,\Omega^{(2)}_2,\Omega^{(1)}_3)
+E(\Omega^{(1)}_1,\Omega^{(1)}_2,\Omega^{(2)}_3)-E(\Omega^{(2)}_1,\Omega^{(2)}_2,\Omega^{(2)}_3)\right|\le 2.
\end{eqnarray}
Using  (\ref{eq:E3spin}) the maximum violation occurs when the l.h.s equals 4. This value can be attained by taking e.g. 
(again the angles $\theta$ do not matter here)
\begin{eqnarray*}
\psi^{(1)}_1=\psi^{(1)}_2=\psi^{(1)}_3=\pi/2, \qquad \varphi^{(1)}_1=\varphi^{(1)}_2=\varphi^{(1)}_3=5\pi/6,\nonumber\\
\psi^{(2)}_1=\psi^{(2)}_2=\psi^{(2)}_3=\pi/2, \qquad \varphi^{(2)}_1=\varphi^{(2)}_2=\varphi^{(2)}_3=\pi/3.
\end{eqnarray*}


\subsection{Optical tomography}\label{vioot}

The tomogram of the state (\ref{eq:psi1}) accordingly to (\ref{eq:opttom}) is given by
\begin{eqnarray}\label{eq:opttomexp}
    p(\gr{X}, \gr{\theta}) = \frac{1}{2\sqrt{\pi^n}}
    \left[ 1+
    2^n\prod_{i=1}^n (X_i^2) 
    +2^{(n+2)/2}\prod_{i=1}^n (X_i) 
    \cos(\theta_1+\ldots+\theta_n) \right] 
    \exp\left[-\sum_{i=1}^n X_i^2\right],
\end{eqnarray}
where $\gr{X} = (X_1, \ldots, X_n)$ and $\gr{\theta}=(\theta_1, \ldots, \theta_n)$. 

We take the sets $Y_k$ and $Z_k$ of Definition \ref{def:setsYZ} to be
\begin{equation*}
    \quad Y_k = [x, +\infty), Z_k = (-\infty, x).
\end{equation*}
For such sets and tomogram (\ref{eq:opttomexp}), 
the correlation function (\ref{eq:Edef}) results 
\begin{eqnarray}\label{eq:Et}
    E(\gr{\theta}) &= 2^{n-1}\left[ {\sf a}_{0}^n(x)+{\sf a}_{1}^n(x)\right] 
+2^n {\sf b}_{0}^n(x)\cos(\theta_1+\ldots+\theta_n),
\end{eqnarray}
where 
\begin{eqnarray*}
    {\sf a}_0(x) = -\frac{1}{2}\erf(x),\;\;\;
    {\sf a}_1(x) = -\frac{1}{2}\erf(x)+\frac{1}{\sqrt{\pi}} xe^{-x^2},\;\;\;
    {\sf b}_0(x) = \frac{1}{\sqrt{2\pi}}e^{-x^2}.
\end{eqnarray*}

We have now to insert (\ref{eq:Et}) 
into (\ref{eq:Mo}) or (\ref{eq:Me}) to get an explicit version fo the Bell inequality. 
In doing so we use a Lemma, reported in \ref{maxval}, showing that the maximal value of 
\begin{equation*}
    \sum\limits^2_{j_1, \ldots, j_n = 1} a_{j_1, \ldots, j_n}
    \cos(\theta^{(j_1)}_1 + \ldots + \theta^{(j_n)}_n)
\end{equation*}
does not exceed $2^{n+(n-1)/2}$ and this value is attained with coefficients of (\ref{eq:Mo}) or (\ref{eq:Me}).
It then follows that the maximal value $f_{n}(x)$ of the l.h.s. of  
(\ref{eq:Mo}) and of (\ref{eq:Me}) is
\begin{equation}\label{fnk}
    f_{n}(x) = \left\{\begin{array}{ccc}
         2^{n}|{\sf a}^n_{0}(x)+{\sf a}^n_{1}(x)| + 
    2^{n+(n+1)/2}| {\sf b}^n_{0}(x)|, & & n \; {\rm odd}\\
    2^{n}|{\sf a}^n_{0}(x)+{\sf a}^n_{1}(x)| + 
    2^{n+n/2} \quad\;\; | {\sf b}^n_{0}(x)|, & & n \; {\rm even}
    \end{array}\right. .
\end{equation}
Figure \ref{fig:v} illustrates the function $f_{n}(x)$ for $n=2,3$. A tiny violation of Bell inequality only occurs 
for $n=3$.

\begin{figure}
\begin{center}
    \includegraphics[scale=1]{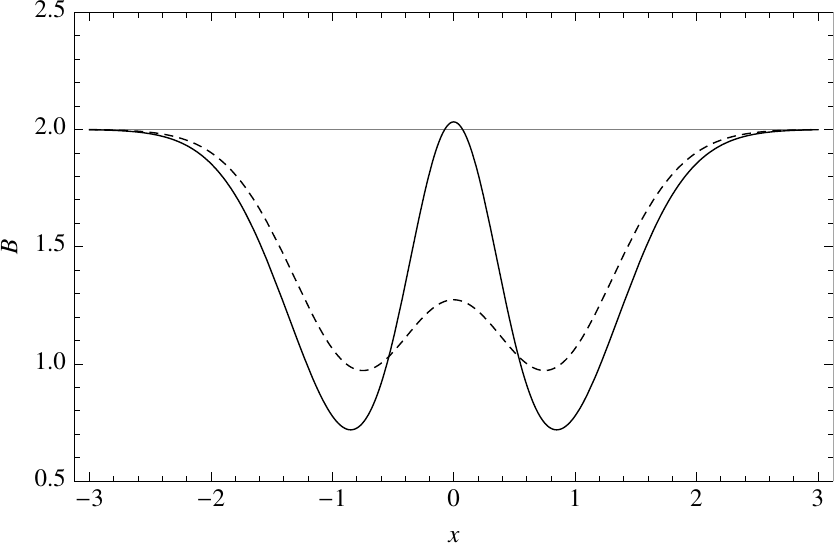}
\end{center}
\caption{Function $f_{n}$ of Eq.~(\ref{fnk}) versus $x$ for  $n=2$ (dashed line) and $n=3$ (solid line). }
\label{fig:v}
\end{figure}


\subsection{Photon-Number tomography}\label{viopnt}

Considering the state (\ref{eq:psi1}) 
its number tomogram (\ref{eq:pntom}) can be computed as
\begin{eqnarray}\label{eq:pntomex}
    p(m_1,\ldots m_n, \alpha_1, \ldots \alpha_n) &=\prod_{i=1}^n 
\frac{|\alpha_i|^{2m_i-2}}{m_i!}    e^{-|\alpha_i|^2} 
 \left| \prod_{i=1}^n \alpha_i +\prod_{i=1}^n (m_i - |\alpha_i|^2)\right|^2.
\end{eqnarray}
We further choose the sets of Definition \ref{def:setsYZ} as
$Z_1 = \ldots = Z_n = \{0\}$, $Y_1 =\ldots = Y_n = \{1,2,3, \ldots\}$.

The corresponding correlation function (\ref{eq:Edef}) for $n=2$ is
\begin{eqnarray}\label{eq:E2pn} 
E(\alpha_1, \alpha_2)=e^{-|\alpha_1|^2-|\alpha_2|^2}&& 
\left[2+4 \Re(\alpha_1 \alpha_2)+2|\alpha_1|^2|\alpha_2|^2
\right.\nonumber\\
&&\left.-\left(1+|\alpha_2|^2\right) e^{|\alpha_1|^2}
-\left(1+|\alpha_1|^2\right) e^{|\alpha_2|^2}
+e^{|\alpha_1|^2+|\alpha_2|^2}\right].
\end{eqnarray}
Furthermore, the Bell inequality for the number tomogram with $n=2$ is from (\ref{eq:Me})
\begin{eqnarray}\label{eq:Bellpn2}
    \Bigl|E(\alpha^{(1)}_1, &\alpha^{(1)}_2) + E(\alpha^{(1)}_1, 
\alpha^{(2)}_2) + 
    E(\alpha^{(2)}_1, \alpha^{(1)}_2) - E(\alpha^{(2)}_1, 
\alpha^{(2)}_2)\Bigr| \leqslant 2,
\end{eqnarray}
for all $\alpha^{(j)}_1, \alpha^{(j)}_2 \in \mathbb{C}$, $j = 1, 2$. 
Figure \ref{fig:pnt} illustrates that this inequality can be violated using (\ref{eq:E2pn}).

Analogously, from (\ref{eq:pntomex}) it follows that the 
correlation function (\ref{eq:Edef}) for $n=3$ is
\begin{eqnarray}\label{eq:E3pn} 
E(\alpha_1, \alpha_2,\alpha_3)=e^{-|\alpha_1|^2-|\alpha_2|^2-|\alpha_3|^2}&& 
\left[-4+8 \Re(\alpha_1 \alpha_2 \alpha_3)-4|\alpha_1|^2 |\alpha_2|^2 |\alpha_3|^2
\right.\nonumber\\
&&\left.+2\left(e^{|\alpha_1|^2}+e^{|\alpha_2|^2}+e^{|\alpha_3|^2}\right)
+2 |\alpha_2|^2 |\alpha_3|^2 e^{|\alpha_1|^2} \right.\nonumber\\
&&\left.+2 |\alpha_1|^2 |\alpha_2|^2 e^{|\alpha_3|^2} +2 |\alpha_1|^2 |\alpha_3|^2 e^{|\alpha_2|^2} 
\right.\nonumber\\
&&\left.-\left(1+|\alpha_1|^2 \right) e^{|\alpha_2|^2+|\alpha_3|^2}
-\left(1+|\alpha_2|^2\right) e^{|\alpha_1|^2+|\alpha_3|^2}
\right.\nonumber\\
&&\left. -\left(1+|\alpha_3|^2\right) e^{|\alpha_1|^2+|\alpha_2|^2}
+e^{|\alpha_1|^2+|\alpha_2|^2+|\alpha_3|^2}\right].
\end{eqnarray}
This time the Bell inequality for the number tomogram reads from (\ref{eq:Mo})
\begin{eqnarray}\label{eq:Bellpn3}
    \Bigl|E(\alpha^{(1)}_1, \alpha^{(1)}_2,  \alpha^{(2)}_3) + E(\alpha^{(1)}_1, \alpha^{(2)}_2,  \alpha^{(1)}_2) 
   + E(\alpha^{(2)}_1, \alpha^{(1)}_2,  \alpha^{(1)}_2) 
   - E(\alpha^{(2)}_1, \alpha^{(2)}_2,  \alpha^{(2)}_2)\Bigr| \leqslant 2.
\end{eqnarray}
This inequality, by numerical checking, results never violated with (\ref{eq:E3pn}) and an example of the behavior of the l.h.s. is shown in  figure \ref{fig:pnt}.

By also choosing $Z_1 = \ldots = Z_n = \{0,\ldots,m\}$, $Y_1 =\ldots = Y_n = \{m+1,m+2, \ldots\}$, 
with $m>0$, neither (\ref{eq:Bellpn2}) nor (\ref{eq:Bellpn3}) will result (by numerical checking) 
ever violated by using (\ref{eq:E2pn}) and (\ref{eq:E3pn}) respectively.

\begin{figure}
\begin{center}
    \includegraphics[scale=1]{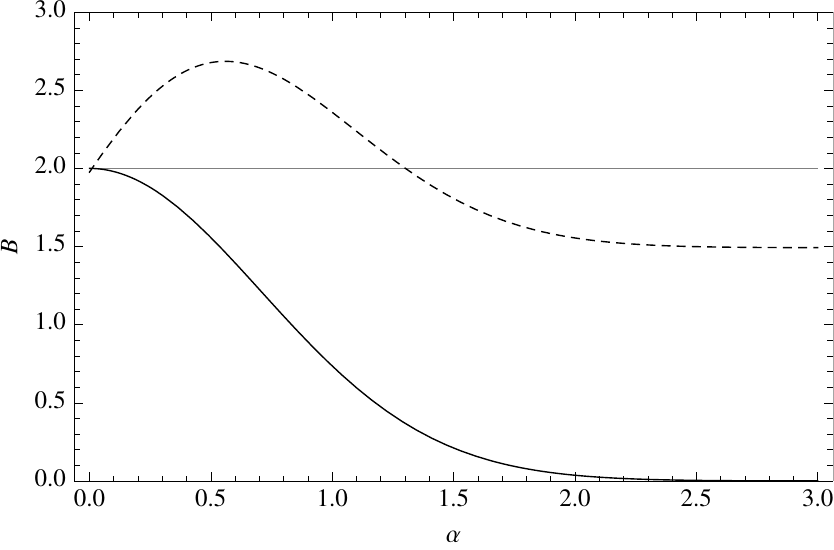}
\end{center}
\caption{The left hand side of (\ref{eq:Bellpn2}) as a function of 
$\alpha^{(2)}_2$ (dashed line); the other parameters are given by $\alpha^{(1)}_1 = 0.165$, 
$\alpha^{(1)}_2 = -0.165$, $\alpha^{(2)}_1 = -0.559$.
The left hand side of (\ref{eq:Bellpn3}) as a function of 
$\alpha^{(2)}_2$ (dashed line); 
the other parameters are given by
$\alpha^{(1)}_1 = \alpha^{(1)}_2 = 0$, $\alpha^{(1)}_3 = 5.936$,
$\alpha^{(2)}_1 4.767$, $\alpha^{(2)}_3 = 4$.}
\label{fig:pnt}
\end{figure}


\section{Concluding remarks}
\label{conclu}

As we have seen from the previous examples the use of finite (namely $2^n$)  
number of  tomograms within a tomographic realization may lead to the evidence of 
nonlocality.  
Actually, it results that finite dimensional systems by means of spin tomograms 
allow for the best evidence of nonlocality. 
In contrast, violations of Bell inequalities seem much harder to uncover
in infinite dimensional systems where ${\cal H}=L_2(\mathbb{R})$.  
Given that we have considered in both cases the same (entangled) state (\ref{eq:psi1}), this difference,
according to Ref. \cite{SAM}, must be ascribed to the diversity of observables employed (from which the tomograms stem).
However, we argue that also the way the spectrum of an observable is binned 
could play a role. 
As matter of fact the choices made in Sections \ref{vioot} and \ref{viopnt} 
for $Y_k$ and $Z_k$ do not exhaust all possibilities of these measurable sets. 
Unfortunately looking at Bell inequalities violations 
using optical tomograms (resp. photon number tomograms) by scanning the possible  
sets $Y_k$ and $Z_k$ appears a daunting task.

All in all the advantage of the tomographic approach is to allow to 
to find the large violations of Bell inequalities typical of spin systems 
also in infinite dimensional systems. In fact, introducing in $L_2(\mathbb{R})^{\otimes n}$
the following local pseudo-spin operators \cite{MISTA}
\begin{eqnarray*}
    \hat{S}^{(k)}_x &=& \sum\limits^{+\infty}_{n_k=0}
    \Bigl( |2n_k\rangle\langle2n_k+1| + |2n_k+1\rangle\langle 2n_k|
    \Bigr), \nonumber\\\
    \hat{S}^{(k)}_y &=& -i\sum\limits^{+\infty}_{n_k=0}
    \Bigl( |2n_k\rangle\langle2n_k+1| - |2n_k+1\rangle\langle 2n_k|
    \Bigr), \nonumber\\\
    \hat{S}^{(k)}_z &=& \sum\limits^{+\infty}_{n_k=0}
    (-1)^{n_k}|n_k\rangle\langle n_k|,
\end{eqnarray*}
where $|n_k\rangle$ are Fock states of the $k$th subsystem, 
we can derive the tomograms of the spin tomography 
realized with the above operators from those of 
 any other tomographic scheme (see e.g. \cite{QSO97}).
The price one ought to pay in such a 
case is the \emph{completeness} of the set of starting tomograms, 
(i.e. a number of tomograms much greater than $2^n$). 


\section*{Acknowledgments}
This work was planned some years ago after an interesting discussion with V. I. Man'ko. 
We affectionately dedicate its completion to him in occasion of his 75th birthday.  



\appendix

\section{}
\label{maxval}

\begin{lemma}
For any coefficients $a_{\gr{j}} = a_{j_1, \ldots, j_n}$ ($\gr{j} = (j_1, 
\ldots, j_n)$) of (\ref{eq:a}) and for any angles $\theta^{(1)}_k$, 
$\theta^{(2)}_k$ ($k = 1, \ldots, n$) we have
\begin{equation}\label{eq:sin}
    \frac{1}{2^n}\left|\sum^2_{\gr{j}=1} a_{\gr{j}}
    \cos\left(\theta^{(j_1)}_1 + \ldots + 
\theta^{(j_n)}_n\right)\right| \leqslant
    2^{(n-1)/2}.
\end{equation}
The equality is attained with coefficients from (\ref{eq:Mo}), (\ref{eq:Me}).
\end{lemma}

\begin{proof} To estimate the l.h.s. of (\ref{eq:sin}) note that
\begin{eqnarray}\label{eq:exp}
    \left|\sum^2_{\gr{j}=1} a_{\gr{j}}
    \cos\left(\theta^{(j_1)}_1 + \ldots +
    \theta^{(j_n)}_n\right)\right| 
    \leqslant  \left|\sum^2_{\gr{j}=1} a_{\gr{j}}
    e^{i\left(\theta^{(j_1)}_1 + \ldots +
    \theta^{(j_n)}_n\right)}\right|,
\end{eqnarray}
so we need to estimate the last sum. To this end we use (\ref{eq:a}) obtaining
\begin{eqnarray}\label{eq:prod}
    \left|\sum^2_{\gr{j}=1} a_{\gr{j}}
    e^{i\left(\theta^{(j_1)}_1 + \ldots +
    \theta^{(j_n)}_n\right)}\right|  
    =\sum_{\varepsilon_1, \ldots, \varepsilon_n = \pm 1}
    c(\varepsilon_1, \ldots, \varepsilon_n)
    \prod^n_{k=1}\Bigl(e^{i\theta^{(1)}_k}+\varepsilon_k 
e^{i\theta^{(2)}_k}\Bigr).
\end{eqnarray}
Next we define
\begin{equation}\label{eq:tp}
    \theta_k = \frac{\theta^{(1)}_k - \theta^{(2)}_k}{2}, \quad
    \varphi_k = \frac{\theta^{(1)}_k + \theta^{(2)}_k}{2}.
\end{equation}
so that the r.h.s. of
(\ref{eq:prod}) simplifies to
\begin{eqnarray*}
    2^n e^{i(\varphi_1+\ldots+\varphi_n)}
    \sum_{\varepsilon_1, \ldots, \varepsilon_n = \pm 1}
    c(\varepsilon_1, \ldots, \varepsilon_n) \prod^n_{k=1} 
a_k(\varepsilon_k),
\end{eqnarray*}
where $a_k(+1) = \cos\theta_k$ and $a_k(-1) = i\sin\theta_k$.
Taking into account that we use absolute value in (\ref{eq:exp}) and divide by 
$2^n$ in (\ref{eq:sin}) we have to prove the following inequality
\begin{equation}\label{eq:ca}
    \left|\sum_{\varepsilon_1, \ldots, \varepsilon_n = \pm 1}
    c(\varepsilon_1, \ldots, \varepsilon_n) \prod^n_{k=1} 
a_k(\varepsilon_k) \right|
    \leqslant 2^{(n-1)/2},
\end{equation}
for any $\pm 1$-valued function $c(\varepsilon_1, \ldots, \varepsilon_n)$. We employ
the induction method. For $n=1$ we simply have
\begin{equation*}
    \Bigl|c(+1)\cos\theta_1 + c(-1)i\sin\theta_1\Bigr| = 1 =
    2^{(1-1)/2}.
\end{equation*}
Then, we can write the sum in (\ref{eq:ca}) as
\begin{eqnarray*}
    \sum_{\varepsilon_1, \ldots, \varepsilon_n = \pm 1}
    c(\varepsilon_1, \ldots, \varepsilon_n) \prod^n_{k=1} a_k(\varepsilon_k) 
     &\equiv& A_{n-1}\cos\theta_n + iB_{n-1}\sin\theta_n,  \nonumber\\
    &=&\sum\limits_{\varepsilon_1, \ldots, \varepsilon_{n-1} = \pm 1}
    c(\varepsilon_1, \ldots, \varepsilon_{n-1}, +1) \prod^{n-1}_{k=1}
    a_k(\varepsilon_k)\cos\theta_n   \nonumber \\
    && +i\sum\limits_{\varepsilon_1, \ldots, \varepsilon_{n-1} = \pm 1}
    c(\varepsilon_1, \ldots, \varepsilon_{n-1}, -1) \prod^{n-1}_{k=1}
    a_k(\varepsilon_k) \sin\theta_n 
\end{eqnarray*}
where, according to the induction assumption, we have
\begin{equation}
    |A_{n-1}|, \; |B_{n-1}| \leqslant 2^{(n-2)/2}.
\end{equation}
The sum in (\ref{eq:ca}) can be estimated in the following way
\begin{eqnarray*}
    \left|\sum_{\varepsilon_1, \ldots, \varepsilon_n = \pm 1}
    c(\varepsilon_1, \ldots, \varepsilon_n) \prod^n_{k=1} 
a_k(\varepsilon_k)
    \right| &=& 
    \Bigl| A_{n-1}\cos\theta_n + iB_{n-1}\sin\theta_n \Bigr|
    \nonumber\\
    & \leqslant &  \sqrt{|A_{n-1}|^2+|B_{n-1}|^2} \leqslant 
    2^{(n-1)/2}.
\end{eqnarray*}

Now we show that with the coefficients of (\ref{eq:Mo}) or (\ref{eq:Me}) the maximal value $2^{(n-1)/2}$ is 
attained. Due to (\ref{eq:exp}) we need to estimate the sum
\begin{equation}\label{eq:Ms}
    S_n = \frac{1}{2^{n+(n-1)/2}}\sum^2_{j_1, \ldots, j_n=1} a_{j_1, 
\ldots, j_n}
    e^{i\left(\theta^{(j_1)}_1 + \ldots + \theta^{(j_n)}_n\right)}
\end{equation}
and show that it can be equal to one by absolute value. First, let us consider 
the case of an odd $n$. From
(\ref{eq:Me}) we have
\begin{equation*}
    a_{j_1, \ldots, j_n} = 2^{(n+1)/2} (-1)^{\delta(j_1, \ldots, 
j_n)}, \quad\quad (j_1, \ldots, j_n) \in J.
\end{equation*}
Furthermore, from (\ref{eq:Mn}) it is 
\begin{equation*}
    S_n = \frac{1}{2^n i} 
    \left[\prod^n_{k=1}\left(e^{i \theta^{(1)}}_k + i e^{i 
\theta^{(2)}}_k\right) - \prod^n_{k=1}\left(e^{i \theta^{(1)}}_k - i 
e^{i \theta^{(2)}}_k\right)\right].
\end{equation*}
Taking into account that each term in these products can be written as
\begin{equation*}
    e^{i \theta^{(1)}}_k \pm e^{i \tilde{\theta}^{(2)}}_k, \quad
    \tilde{\theta}^{(2)}_k = \theta^{(2)}_k + \pi/2,
\end{equation*}
and using the relations (\ref{eq:tp}), $S_n$ can be simplified to
\begin{equation*}
    S_n = \frac{1}{i}\left(\prod^n_{k=1}\cos\theta^\prime_k \pm i
    \prod^n_{k=1}\sin\theta^\prime_k\right)\,e^{i(\varphi^\prime_1 + 
\ldots + \varphi^\prime_n)},
\end{equation*}
where $\theta^\prime_k = \theta_k - \pi/4$, 
$\varphi^\prime_k = \varphi_k + \pi/4$.
It is clear that the imaginary part of the sum $S_n$ takes its maximal absolute 
value $1$ when, for example, $\theta_k = \varphi_k = 0$ for $k = 1, \ldots, n$.

Now we consider the case of an even $n$. 
The coefficients $a_{j_1, \ldots, j_n}$ 
in this case come from (\ref{eq:Mo}) 
\begin{equation*}
    a_{j_1, \ldots, j_n} = 2^{n/2} (-1)^{\tilde{\delta}(j_1, \ldots, j_n)},
\end{equation*}
and the sum $S_n$ (\ref{eq:Ms}) becomes
\begin{eqnarray*}
    S_n &= \frac{1}{i2^n\sqrt{2}}
    (e^{i \theta^{(1)}}_n + e^{i \theta^{(2)}}_n) 
    \left[\prod^{n-1}_{k=1}\left(e^{i \theta^{(1)}}_k + e^{i 
\tilde{\theta}^{(2)}}_k\right) - \prod^{n-1}_{k=1}\left(e^{i 
\theta^{(1)}}_k - e^{i \tilde{\theta}^{(2)}}_k\right)\right]\\
&+  \frac{1}{i2^n\sqrt{2}} (e^{i \theta^{(1)}}_n e^{i \theta^{(2)}}_n)
\left[\prod^{n-1}_{k=1}\left(e^{i \tilde{\theta}^{(1)}}_k + 
e^{i \theta^{(2)}}_k\right) - \prod^{n-1}_{k=1}\left(e^{i 
\tilde{\theta}^{(1)}}_k - e^{i \theta^{(2)}}_k\right)
\right].
\end{eqnarray*}
According to (\ref{eq:tp}) $S_n$ can be simplified to 
\begin{eqnarray*}
    S_n &= \frac{e^{i\varphi}}{\sqrt{2}}
    \left[\left(i\prod^{n-1}_{k=1}\cos\theta^\prime_k \mp 
    \prod^{n-1}_{k=1}\sin\theta^\prime_k\right)\cos\theta_n 
   +\left(\prod^{n-1}_{k=1}\cos\theta^{\prime\prime}_k \pm 
i\prod^{n-1}_{k=1}\sin\theta^{\prime\prime}_k\right)\sin\theta_n\right],
\end{eqnarray*}
where $\theta^\prime_k = \theta_k - \pi/4$, $\theta^{\prime\prime}_k 
= \theta_k + \pi/4$ and
$\varphi = \varphi_1 + \ldots \varphi_n + (n-1) \pi/4$.
The imaginary part of $S_n$ is (when $\varphi = 0$)
\begin{eqnarray*}
    |\im(S_n)| &= \frac{1}{\sqrt{2}}
    \left| \cos\theta_n \prod^{n-1}_{k=1}\cos(\theta_k-\pi/4)
    \pm 
      \sin\theta_n \prod^{n-1}_{k=1}\sin(\theta_k+\pi/4)
    \right|.
\end{eqnarray*}
It is clear that this expression takes its maximal value $1$ when $\theta_k = 
\pi/4$, $k = 1, \ldots, n-1$, and $\theta_n = \pm \pi/4$. This completes the 
proof.
\end{proof}


\end{document}